\documentclass[11pt]{article}
\usepackage{amsmath,amsxtra}
\usepackage{amscd}
\usepackage{graphicx}
\usepackage{latexsym}
\usepackage{amsmath,amsthm,amsfonts,amssymb,cite}

\usepackage{latexsym}
\usepackage{amsmath}

\newtheorem{pro}{Proposition}
\theoremstyle{remark}

\newtheorem{rem}{Remark}
\numberwithin{equation}{section}
\begin{document}
\hoffset = -2.4truecm \voffset = -2truecm
\renewcommand{\baselinestretch}{1.2}
\newcommand{\mb}{\makebox[10cm]{}\\ }
\date{}

\newtheorem{theorem}{Theorem}
\newtheorem{proposition}{Proposition}
\newtheorem{lemma}{Lemma}
\newtheorem{definition}{Definition}


\title{An insight into the $q$-difference two-dimensional Toda lattice equation, $q$-difference sine-Gordon equation and their integrability }
\author{C.X. Li$^{1,\dagger}$, H.Y. Wang$^{2}$, Y.Q. Yao$^3$ and S.F. Shen$^4$\\
$^1$School of Mathematical Sciences,
Capital Normal University,
    Beijing 100048, China\\
$^2$School of Mathematics, 
Renmin University of China,
Beijing 100872, China\\
$^3$College of Science, China Agricultural University, 
Beijing 100083, China\\
$^4$Department of Applied Mathematics, Zhejiang University of Technology,
Hangzhou 310023, China}
\date{}
 \maketitle
\begin{abstract}
In our previous work \cite{LNS}, we constructed quasi-Casoratian solutions to the noncommutative $q$-difference two-dimensional Toda lattice ($q$-2DTL) equation by Darboux transformation, which we can prove produces the existing Casoratian solutions to the bilinear $q$-2DTL equation obtained by Hirota's bilinear method in commutative setting. It is actually true that one can not only construct solutions to soliton equations but also solutions to their corresponding B$\ddot{a}$cklund transformations by their Darboux transformations and binary Darboux transformations. To be more specific, eigenfunctions produced by iterating Darboux transformations and binary Darboux transformations for soliton equations give nothing but determinant solutions to their B$\ddot{a}$cklund transformations, individually. This reveals the profound connections between Darboux transformations and Hirota's bilinear method. In this paper, we shall expound this viewpoint in the case of the $q$-2DTL equation. First, we derive a generalized bilinear B$\ddot{a}$cklund transformation and thus a generalized Lax pair for the bilinear $q$-2DTL equation. And then we successfully construct the binary Darboux transformation for the $q$-2DTL equation, based on which, Grammian solutions expressed in terms of quantum integrals are established for both the bilinear $q$-2DTL equation and its bilinear B$\ddot{a}$cklund transformation. In the end, by imposing the 2-periodic reductions on the corresponding results of the $q$-2DTL equation, we derive a $q$-difference sine-Gordon equation, a modified $q$-difference sine-Gordon equation and obtain their corresponding solutions.
\end{abstract}

\section{Introduction}
With the discovery of quantum groups \cite{DVG,JM}, the studies of integrable systems have entered a new phase. Some kinds of $q$-special functions naturally appear in the representation theory of quantum groups, where the $q$-difference operator plays a role similar to the differential operator in the theory of special functions \cite{EH,KTH,MMNN}. As is well known, soliton equations have a close relationship with special functions\cite{AN,KO}. Therefore, the study of $q$-difference integrable systems arises as an interesting subject. 

As two powerful tools, both Hirota's bilinear method and Darboux transformation have been extensively used to study continuous or discrete integrable systems, particularly, their solutions. On one hand, one can derive solutions to soliton equations and their corresponding B$\ddot{a}$cklund transformations by Hirota's bilinear method \cite{HR1,HR2}. On the other hand, it is well known that solutions to soliton equations can be constructed by their Darboux transformations and binary Darboux transformations \cite{DT1,DT2,DT3}. In addition, it is remarkable that Darboux transformations and binary Darboux transformations for soliton equations can produce solutions to their corresponding B$\ddot{a}$cklund transformations too. Specifically speaking, eigenfunctions produced by iterating Darboux transformations and binary Darboux transformations for soliton equations provide solutions to their B$\ddot{a}$cklund transformations, individually. These results obtained by the two methods are amazingly consistent with each other, which reveals the profound connections between Hirota's bilinear method and Darboux transformation. This viewpoint could be observed in \cite{NW} where the Grammian solutions to the bilinear two-dimensional Toda lattice equation and its bilinear B$\ddot{a}$cklund transformations were constructed by the binary Darboux transformation. We shall expound this viewpoint and apply it to $q$-difference integrable systems explicitly in this paper.  

$q$-Difference integrable systems can be regarded as a kind of discretized systems which reduce to classical integrable systems under the continuum limit $q\rightarrow 1$ \cite{VKC}. The study of $q$-analogues of classical integrable systems in parallel with classical integrable systems has attracted much attention\cite{PNGR,HS,TSL,TL,HLC,KOS}. In this paper, we will focus on the $q$-2DTL equation which is a $q$-difference version of the two-dimensional Toda lattice ($2$DTL) equation. In \cite{KOS}, a bilinear $q$-2DTL equation together with its Casoratian solutions were first proposed. As the reduction of the bilinear $q$-2DTL equation, a $q$-cylindrical Toda lattice equation was derived which admmited Casorati determinant solutions with entries given by the $q$-Bessel functions. Recently, we presented a slightly different nonlinear $q$-2DTL equation whose bilinear form is the same as the one obtained in \cite{KOS}. We constructed its bilinear B$\ddot{a}$cklund transformations and Lax pair by using Hirota's bilinear method. Besides, Darboux transformation was established to construct its quasi-Casoratian solutions for a noncommutative $q$-2DTL equation\cite{LNS}. Actually, we can prove that the same Darboux transformation holds true for the commutative $q$-2DTL equation which can be used to reconstruct the existing Casoratian solutions to the bilinear $q$-2DTL equation and its bilinear B$\ddot{a}$cklund transformation. Furthermore, matrix integral solutions to the $q$-2DTL equation and its Pfaffianized system were presented in \cite {LQS}. However, how to construct the binary Darboux transformation for the $q$-2DTL equation and further derive its Grammian solutions still remain a challenging problem. To go a step further, periodic reductions on the results of the $q$-2DTL equation are of great interest as well. We will tell the whole story in this paper. 

The paper is organized as follows. We give a brief review of some known results on the $q$-2DTL equation and present its generalized bilinear B$\ddot{a}$cklund transformation and Lax pair in Section 2. In Section 3, we will explain how to derive existing Casorati solutions to the bilinear B$\ddot{a}$cklund transformation for the q-$2$DTL equation from the Darboux transformation. In Section 4, the binary Darboux transformation is successfully constructed for the $q$-2DTL equation, based on which Grammian solutions expressed in terms of quantum integrals are established for both the $q$-2DTL equation and its B$\ddot{a}$cklund transformations. Section 5 is devoted to the $2$-periodic reductions of the corresponding results of the $q$-2DTL equation. As a result, the $q$-difference sine-Gordon equation and its modified system are proposed. Meanwhile, their solutions are presented, respectively. Concluding remarks are given in Section 6.

\section{A generalized Lax pair for the $q$-2DTL equation}
The nonlinear q-$2$DTL equation was first proposed in \cite{KOS}.  Later on, a slightly different nonlinear $q$-2DTL equation was presented in \cite{LNS}. These two nonlinear equations correspond to the same bilinear equation under different dependent variable transformations. In this paper, we will adopt the nonlinear $q$-2DTL equation appearing in the latter paper. 

The $q$-2DTL equation considered reads as
\begin{eqnarray}
	&&D_{q^ \alpha, x}V_n(x,y)=J_{n+1}(x,q^ \beta y)V_n(x,y)-V_n(q^ \alpha x,y)J_n(x,y),\label{E1}\\
	&&D_{q^ \beta, y}J_n(x,y)=V_n(x,y)-V_{n-1}(q^ \alpha x,y),\label{E2}
\end{eqnarray}
where the $q$-difference operator (or the Jackson derivative) is defined by
\begin{eqnarray*}
	&&D_{q^ \alpha, x}(f(x,y))=\frac{\sigma_1(f(x,y))-f(x,y)}{(q-1)x},\\
	&&D_{q^ \beta, y}(f(x,y))=\frac{\sigma_2(f(x,y))-f(x,y)}{(q-1)y}
\end{eqnarray*}
with the $q$-shift operator given by $\sigma_1 (f(x,y))=f(q^ \alpha x,y)$ and $\sigma_2(f(x,y))=f(x,q^\beta y)$. For the sake of simplicity, we denote $D_1=D_{q^\alpha,x}$, $D_2=D_{q^\beta,y}$ and $f_n=f_n(x,y)$ for any function $f_n(x,y)$ without any ambiguity from now on.  

By introducing $V_n=\sigma_2(X_{n+1})X_n^{-1}$ and $J_n=D_1(X_n)X_n^{-1}$, eqs. \eqref{E1}$\sim$\eqref{E2} can be reduced into a single equation 
\begin{equation}
	D_2(D_1(X_n)X_n^{-1})=\sigma_2(X_{n+1})X_n^{-1}-\sigma_1(\sigma_2(X_n)X_{n-1}^{-1}). \label{cbQT}
\end{equation}
By assuming $X_n=\tau_n\tau_{n-1}^{-1}$, \eqref{cbQT} can be further transformed into the bilinear $q$-2DTL equation \cite{KOS}
	\begin{equation}
		D_1(D_2(\tau_n))\tau_n-D_1(\tau_n)D_2(\tau_n)=\sigma_2(\tau_{n+1})\sigma_1(\tau_{n-1})-\sigma_2(\sigma_1(\tau_n))\tau_n \label{qTL}.
	\end{equation}

Denote $a(x)=(q-1)x$ and $b(y)=(q-1)y$. In \cite{LNS}, a special bilinear B$\ddot{a}$cklund transformation was presented for the $q$-2DTL equation \eqref{qTL}. It was given by 
\begin{eqnarray}
	&&D_1(\tau_n)\tau'_{n-1}-D_1(\tau'_{n-1})\tau_n=\sigma_1(\tau_{n-1})\tau'_n,\label{kBT1}\\
	&&D_2(\tau_n)\tau'_n-D_2(\tau'_n)\tau_n=- \sigma_2(\tau_{n+1})\tau'_{n-1}, \label{kBT2}
\end{eqnarray}
which, by setting $\phi_{n+1}=\tau_n'/\tau_n$, led to the following special Lax pair 
\begin{eqnarray}
	&&D_1\phi_n(x,y)=-\phi_{n+1}(x,y)+J_n(x,y)\phi_n(x,y),\label{kLP1}\\
	&&D_2\phi_n(x,y)=V_{n-1}(x,y)\phi_{n-1}(x,y),\label{kLP2}
\end{eqnarray}	
where 
\begin{align*}
	V_n=\frac{\tau_{n-1}\sigma_2(\tau_{n+1})}{\tau_n\sigma_2(\tau_n)}=\sigma_2(X_{n+1})X_n^{-1},\ J_n=\frac{1}{a(x)}\left(\frac{\tau_{n-1}\sigma_1(\tau_n)}{\tau_n\sigma_1(\tau_{n-1})}-1\right)=D_1(X_n)X_n^{-1}.
\end{align*}

Actually, concerning \eqref{qTL}, we manage to obtain a more general bilinear B$\ddot{a}$cklund transformation and thus a generalized Lax pair which can be stated as 
\begin{pro}
The bilinear $q$-2DTL equation \eqref{qTL} has the bilinear B$\ddot{a}$cklund transformation
\begin{align}
	D_1(\tau_n)\tau'_{n-1}-D_1(\tau'_{n-1})\tau_n&=-\lambda^{-1}\sigma_1(\tau_{n-1})\tau'_n+\nu \tau_n\sigma_1(\tau'_{n-1}),\label{bBT1}\\
	D_2(\tau_n)\tau'_n-D_2(\tau'_n)\tau_n&=\lambda \sigma_2(\tau_{n+1})\tau'_{n-1}-\mu \sigma_2(\tau'_n)\tau_n,\label{bBT2}
\end{align}
where $\lambda$, $\nu$ and $\mu$ are arbitrary constants.
\end{pro}

\begin{proof}
Suppose that $\tau_n$ is a solution of the bilinear $q$-2DTL equation \eqref{qTL}. If we can show that $\tau'_n$ is also a solution of \eqref{qTL}, then \eqref{bBT1} and \eqref{bBT2} form a B$\ddot{a}$cklund transformation. 

It is obvious that \eqref{qTL} has the equivalent form 
\begin{align}
	(1+a(x)b(y))\sigma_1\sigma_2(\tau_n)\tau_n-\sigma_1(\tau_n)\sigma_2(\tau_n)-a(x)b(y)\sigma_2(\tau_{n+1})\sigma_1(\tau_{n-1})=0.\label{ef1}
\end{align}
By virtue of \eqref{bBT1} and \eqref{bBT2}, we can prove
\begin{align*}
	P&\equiv \left[(1+a(x)b(y))\sigma_1\sigma_2(\tau_n)\tau_n-\sigma_1(\tau_n)\sigma_2(\tau_n)-a(x)b(y)\sigma_2(\tau_{n+1})\sigma_1(\tau_{n-1})\right]\sigma_1\sigma_2(\tau'_n)\tau'_n\\
	&\quad-[(1+a(x)b(y))\sigma_1\sigma_2(\tau'_n)\tau'_n-\sigma_1(\tau'_n)\sigma_2(\tau'_n)-a(x)b(y)\sigma_2(\tau'_{n+1})\sigma_1(\tau'_{n-1})]\sigma_1\sigma_2(\tau_n)\tau_n\\
	&=a(x)b(y)[\sigma_2(\tau'_{n+1})\sigma_1(\tau'_{n-1})\sigma_1\sigma_2(\tau_n)\tau_n-\sigma_2(\tau_{n+1})\sigma_1(\tau_{n-1})\sigma_1\sigma_2(\tau'_n)\tau'_n]\\
	&\quad +\sigma_1(\tau'_n)\sigma_2(\tau'_n)\sigma_1\sigma_2(\tau_n)\tau_n-\sigma_1(\tau_n)\sigma_2(\tau_n)\sigma_1\sigma_2(\tau'_n)\tau'_n\\
	&=b(y)[\sigma_1(\tau'_{n-1})\tau_n(\lambda(1+a(x)\nu)\sigma_2(\tau_{n+1})\sigma_2\sigma_1(\tau'_{n})-\lambda\sigma_2\sigma_1(\tau_{n+1})\sigma_2(\tau'_{n}))\\
	&\quad -  \sigma_2(\tau_{n+1})\sigma_1\sigma_2(\tau'_n)(\lambda(1+a(x)\nu)\tau_n\sigma_1(\tau'_{n-1})-\lambda\sigma_1(\tau_n)\tau'_{n-1})]\\
	&\quad +\sigma_1(\tau'_n)\sigma_2(\tau'_n)\sigma_1\sigma_2(\tau_n)\tau_n-\sigma_1(\tau_n)\sigma_2(\tau_n)\sigma_1\sigma_2(\tau'_n)\tau'_n\\
	&=\lambda b(y)[ \sigma_2(\tau_{n+1})\sigma_1\sigma_2(\tau'_n)\sigma_1(\tau_n)\tau'_{n-1}-\sigma_1(\tau'_{n-1})\tau_n\sigma_2\sigma_1(\tau_{n+1})\sigma_2(\tau'_{n})]\\
	&\quad +\sigma_1(\tau'_n)\sigma_2(\tau'_n)\sigma_1\sigma_2(\tau_n)\tau_n-\sigma_1(\tau_n)\sigma_2(\tau_n)\sigma_1\sigma_2(\tau'_n)\tau'_n\\
	&=\sigma_2(\tau'_n)\tau_n\sigma_1[\tau'_n\sigma_2(\tau_n)-\lambda b(y)\tau'_{n-1}\sigma_2(\tau_{n+1})]-\sigma_1(\sigma_2(\tau'_n)\tau_n)[\sigma_2(\tau_n)\tau'_n-\lambda b(y)\sigma_2(\tau_{n+1})\tau'_{n-1}]\\
	&=\sigma_2(\tau'_n)\tau_n [\sigma_1\sigma_2(\tau'_n)\sigma_1(\tau_n)-\mu b(y)\sigma_1\sigma_2(\tau'_n)\sigma_1(\tau_n)]-\sigma_1(\sigma_2(\tau'_n)\tau_n)[\sigma_2(\tau'_n)\tau_n-\mu b(y)\sigma_2(\tau'_n)\tau_n]\\
	&=0.
\end{align*}
In this way, we have completed the proof of Proposition 1 by showing that $\tau'_n$ satisfies the equivalent form \eqref{ef1} of \eqref{qTL}.  
\end{proof}

In the following, we are going to derive a generalized Lax pair for the $q$-2DTL equation \eqref{qTL} from its bilinear B$\ddot{a}$cklund transformation \eqref{bBT1} and \eqref{bBT2}. For this purpose, we set $\phi_{n+1}=\tau'_n/\tau_n$. Then \eqref{bBT1} and \eqref{bBT2} transform into the following Lax pair 
\begin{align}
	D_1(\phi_n)&=\frac{1}{\lambda(1+\nu a(x))}(\phi_{n+1}-\lambda \nu \phi_n+\lambda J_n\phi_n),\label{LP1}\\
	D_2(\phi_n)&=\frac{1}{1-\mu b(y)}(\mu\phi_n-\lambda V_{n-1}\phi_{n-1}).\label{LP2}
\end{align}
It is not difficult to check that the compatibility condition of the Lax pair \eqref{LP1} and \eqref{LP2} gives the nonlinear $q$-2DTL equation \eqref{E1} and \eqref{E2} or equivalently, the bilinear $q$-2DTL equation \eqref{qTL}.

\begin{rem}
By choosing $\nu=\mu=0$ and $\lambda=-1$, one can easily derive the existing B$\ddot{a}$cklund transformation \eqref{kBT1}$\sim$\eqref{kBT2} and Lax pair \eqref{kLP1}$\sim$\eqref{kLP2} from the generalized B$\ddot{a}$cklund transformation \eqref{bBT1}$\sim$\eqref{bBT2} and Lax pair \eqref{LP1}$\sim$\eqref{LP2}, respectively.
\end{rem}

\begin{rem}
Under the continuum limit $q\rightarrow 1$, the bilinear $q$-2DTL equation \eqref{qTL} becomes the well-known bilinear 2DTL equation 
\begin{align}
D_xD_y\tau_n\cdot \tau_n=2(\tau_{n+1}\tau_{n-1}-\tau_n^2) \label{TL}
\end{align}
where the bilinear operator is defined by
\begin{align*}
D_x^mD_t^n f(x,t)\cdot g(x,t)=(\partial_x-\partial_{x'})^m(\partial_t-\partial_{t'})^nf(x,t)g(x',t')|_{x'=x,t'=t}.
\end{align*}
Correspondingly, the bilinear B$\ddot{a}$cklund transformation \eqref{bBT1} and \eqref{bBT2} become 
\begin{align*}
&D_y\tau_n\cdot \tau_n'=\lambda \tau_{n+1}\tau_{n-1}'-\mu \tau_n\tau_n',\\
&D_x \tau_{n}\cdot\tau_{n-1}'=-\lambda^{-1}\tau_{n-1}\tau_n'+\nu \tau_n\tau_{n-1}'
\end{align*}
which are nothing but the bilinear B$\ddot{a}$cklund transformation for the well-known bilinear 2DTL equation \eqref{TL} appearing in \cite{HR1,HR2}. This justifies the rationality of our results.
\end{rem}


\section{Darboux transformation for the $q$-2DTL equation}
In literature, Casoratian solutions for the $q$-2DTL equation \eqref{qTL} and its B$\ddot{a}$cklund transformation \eqref{kBT1}$\sim$\eqref{kBT2} were constructed by using Hirota's bilinear method. In this section, we will explain how to recover these solutions by Darboux transformation for \eqref{cbQT}.

It has been shown that the $q$-2DTL equation has the Casoratian solution \cite{KOS}
\begin{align}
	\tau_n=
	\begin{vmatrix}
		f_n^{(1)}&f_{n+1}^{(1)}&\cdots&f_{n+N-1}^{(1)}\\
		f_n^{(2)}&f_{n+1}^{(2)}&\cdots&f_{n+N-1}^{(2)}\\
		\vdots&\vdots&\cdots&\vdots\\
		f_n^{(N)}&f_{n+1}^{(N)}&\cdots&f_{n+N-1}^{(N)}
	\end{vmatrix},\label{QTBFS}
\end{align}
where $f_n^{(k)}, k=1,\cdots,N,$ satisfy the dispersion relations
\begin{equation}\label{DS}
	D_1f_n^{(k)}=-f_{n+1}^{(k)},\,\ D_2f_n^{(k)}=f_{n-1}^{(k)}.
\end{equation}
In addition, Casoratian solutions to B$\ddot{a}$cklund transformation \eqref{kBT1}$\sim$\eqref{kBT2}  were presented as \cite{LNS}
\begin{align}
\tau_n&=|0,1,\cdots,N-1|,\label{bBS1}\\
\tau_n'&=|0,1,\cdots,N-1,N|,\label{bBS2}
\end{align}
where we have adopted the compact notations in \cite{KOS} for the sake of simplicity.

%
%
\subsection{Casoratian solutions to the bilinear $q$-2DTL equation and its B$\ddot{a}$cklund transformation }
In \cite{LNS}, a noncommutative $q$-2DTL equation was considered whose quasi-Casoratian solutions were constructed by Darboux transformation. Actually the noncommutative $q$-2DTL equation considered is exactly the same as \eqref{cbQT} formally. The only difference lies in that two different functions are generally 
noncommutative in noncommutative setting. Therefore, Darboux transformation for the noncommutative $q$-2DTL equation also holds for the $q$-2DTL equation \eqref{cbQT}. 

Let $\theta_{n,i},\, i=1,\dots,N,N+1$ be a particular set of eigenfunctions of the linear system \eqref{kLP1}$\sim$\eqref{kLP2} and introduce the notation $\Theta_n=(\theta_{n,1},\dots,\theta_{n,N})$. The Darboux transformation, determined by a particular eigenfunction $\theta_n$ for the special Lax pair \eqref{kLP1} and \eqref{kLP2} is 
\begin{align}
	\tilde{\phi}_n&
	=-\phi_{n+1}+\theta_{n+1}\theta_n^{-1}\phi_n,\nonumber\\
	\tilde{X}_n&=-\theta_{n+1}\theta_n^{-1}X_n,\nonumber
\end{align}
and this may be iterated by defining
\begin{align}
	\phi_n[k+1]&=-\phi_{n+1}[k]+\theta_{n+1}[k]\theta_n^{-1}[k]\phi_n[k],\nonumber\\
	X_n[k+1]&=-\theta_{n+1}[k]\theta_n^{-1}[k]X_n[k],\nonumber
\end{align}
where $\phi_n[1]=\phi_n,\, X_n[1]=X_n$ and 
$$\theta_n[k]=\phi_n[k]|_{\phi_n\rightarrow\theta_{n,k}}.$$

In what follows, we will show by induction that the results of $N$-repeated Darboux transformations $\phi_n[N+1]$ and $X_n[N+1]$ can be expressed in closed forms as
\begin{align}
	\phi_n[N+1]&=
	\begin{vmatrix}
		\Theta_{n+N}&\phi_{n+N}\\
		\Theta_{n+N-1}&\phi_{n+N-1}\\
		\vdots&\vdots\\
		\Theta_{n}&\phi_{n}
	\end{vmatrix}\cdot\begin{vmatrix}
		\Theta_{n+N-1}\\
		\vdots\\
		\Theta_{n}
	\end{vmatrix}^{-1}\nonumber\\
	&=(-1)^N\begin{vmatrix}
		\theta_{n,1}&\dots&\theta_{n+N,1}\\
		\vdots&\ddots&\vdots\\
		\theta_{n,N}&\dots&\theta_{n+N,N}\\
		\phi_{n}&\dots&\phi_{n+N}
	\end{vmatrix}\cdot\begin{vmatrix}
		\theta_{n,1}&\dots&\theta_{n+N-1,1}\\
		\vdots&\ddots&\vdots\\
		\theta_{n,N}&\dots&\theta_{n+N-1,N}
	\end{vmatrix}^{-1},\label{EFD1}\\
	X_n[N+1]&=(-1)^N
	\begin{vmatrix}
		\Theta_{n+N}\\
		\vdots\\
		\Theta_{n+1}
	\end{vmatrix}\cdot\begin{vmatrix}
		\Theta_{n+N-1}\\
		\vdots\\
		\Theta_{n}
	\end{vmatrix}^{-1}X_{n}\nonumber\\
	&=(-1)^N\begin{vmatrix}
		\theta_{n+1,1}&\dots&\theta_{n+N,1}\\
		\vdots&\ddots&\vdots\\
		\theta_{n+1,N}&\dots&\theta_{n+N,N}
	\end{vmatrix}\cdot\begin{vmatrix}
		\theta_{n,1}&\dots&\theta_{n+N-1,1}\\
		\vdots&\ddots&\vdots\\
		\theta_{n,N}&\dots&\theta_{n+N-1,N}
	\end{vmatrix}^{-1}X_{n}.\label{EFD2}
\end{align}

The initial case $N=1$ is obviously true for $\phi_n[N+1]$ and $X_n[N+1]$. By using the Jacobi identity, we have
\begin{align*}
	\phi_n[N+2]=&-\phi_{n+1}[N+1]+\theta_{n+1}[N+1]\theta_n[N+1]^{-1}\phi_n[N+1]\\
	=&-\left(\begin{vmatrix}
		\Theta_{n+N+1}&\phi_{n+N+1}\\
		\vdots&\vdots\\
		\Theta_{n+1}&\phi_{n+1}
	\end{vmatrix}\cdot \begin{vmatrix}
		\Theta_{n+N}&\theta_{n+N,N+1}\\
		\vdots&\vdots\\
		\Theta_{n}&\theta_{n,N+1}
	\end{vmatrix}-\begin{vmatrix}
		\Theta_{n+N+1}&\theta_{n+N+1,N+1}\\
		\vdots&\vdots\\
		\Theta_{n+1}&\theta_{n+1,N+1}
	\end{vmatrix}\cdot \begin{vmatrix}
		\Theta_{n+N}&\phi_{n+N}\\
		\vdots&\vdots\\
		\Theta_{n}&\phi_{n}
	\end{vmatrix}\right)\\
	&\cdot\begin{vmatrix}
		\Theta_{n+N}\\
		\vdots\\
		\Theta_{n+1}
	\end{vmatrix}^{-1}\cdot \begin{vmatrix}
		\Theta_{n+N}&\theta_{n+N,N+1}\\
		\vdots&\vdots\\
		\Theta_{n}&\theta_{n,N+1}
	\end{vmatrix}^{-1}\\
	=&\begin{vmatrix}
		\Theta_{n+N+1}&\theta_{n+N+1,N+1}&\phi_{n+N+1}\\
		\vdots&\vdots&\vdots\\
		\Theta_{n}&\theta_{n,N+1}&\phi_n
	\end{vmatrix}\cdot \begin{vmatrix}
		\Theta_{n+N}\\
		\vdots\\
		\Theta_{n+1}
	\end{vmatrix}\cdot\begin{vmatrix}
		\Theta_{n+N}\\
		\vdots\\
		\Theta_{n+1}
	\end{vmatrix}^{-1}\cdot \begin{vmatrix}
		\Theta_{n+N}&\theta_{n+N,N+1}\\
		\vdots&\vdots\\
		\Theta_{n}&\theta_{n,N+1}
	\end{vmatrix}^{-1}\\
	=&\begin{vmatrix}
		\Theta_{n+N+1}&\theta_{n+N+1,N+1}&\phi_{n+N+1}\\
		\vdots&\vdots&\vdots\\
		\Theta_{n}&\theta_{n,N+1}&\phi_n
	\end{vmatrix}\cdot \begin{vmatrix}
		\Theta_{n+N}&\theta_{n+N,N+1}\\
		\vdots&\vdots\\
		\Theta_{n}&\theta_{n,N+1}
	\end{vmatrix}^{-1}
\end{align*}
and 
\begin{align*}
	X_n[N+2]=&-\theta_{n+1}[N+1]\theta_n^{-1}[N+1]X_n[N+1]\\
	&=(-1)^{N+1}\begin{vmatrix}
		\Theta_{n+N+1}&\theta_{n+N+1,N+1}\\
		\vdots&\vdots\\
		\Theta_{n+1}&\theta_{n+1,N+1}
	\end{vmatrix}\cdot \begin{vmatrix}
		\Theta_{n+N}\\
		\vdots\\
		\Theta_{n+1}
	\end{vmatrix}^{-1}\cdot\begin{vmatrix}
		\Theta_{n+N-1}\\
		\vdots\\
		\Theta_{n}
	\end{vmatrix}\cdot\begin{vmatrix}
		\Theta_{n+N}&\theta_{n+N,N+1}\\
		\vdots&\vdots\\
		\Theta_{n}&\theta_{n,N+1}
	\end{vmatrix}^{-1}\\
	&\cdot \begin{vmatrix}
		\Theta_{n+N}\\
		\vdots\\
		\Theta_{n+1}
	\end{vmatrix}\cdot\begin{vmatrix}
		\Theta_{n+N-1}\\
		\vdots\\
		\Theta_{n}
	\end{vmatrix}^{-1}X_{n}\\
	&=(-1)^{N+1}\begin{vmatrix}
		\Theta_{n+N+1}&\theta_{n+N+1,N+1}\\
		\vdots&\vdots\\
		\Theta_{n+1}&\theta_{n+1,N+1}
	\end{vmatrix}\cdot\begin{vmatrix}
		\Theta_{n+N}&\theta_{n+N,N+1}\\
		\vdots&\vdots\\
		\Theta_{n}&\theta_{n,N+1}
	\end{vmatrix}^{-1}X_{n}.
\end{align*}
These prove the inductive steps for both $\phi_n[N+1]$ and $X_n[N+1]$ and thus the proof is completed. 

Notice the transformations $X_{n+1}=\tau_{n+1}\tau_{n}^{-1}$ and $\phi_{n+1}=\tau_n'\tau_n^{-1}$. From the expressions for $\phi_n[N+1]$ and $X_n[N+1]$ given by \eqref{EFD1} and \eqref{EFD2}, one can easily derive solutions to the bilinear $q$-2DTL equation \eqref{qTL} and its B$\ddot a$cklund transformation \eqref{kBT1}$\sim$\eqref{kBT2}. By taking the seed solution $X_n=1$, we have 
 \begin{align}
 X_n[N+1]=\tau_{n+1}\tau_n^{-1},\, \phi_n[N+1]=\tau_n'\tau_n^{-1}
 \end{align}
 with
 \begin{equation}\label{DTS}
	\tau_n=
	\begin{vmatrix}
		\theta_{n,1}&\cdots&\theta_{n+N-1,1}\\
		\theta_{n,2}&\cdots&\theta_{n+N-1,2}\\
		\vdots&\cdots&\vdots\\
		\theta_{n,N}&\cdots&\theta_{n+N-1,N}
	\end{vmatrix},\, \,\tau_n'=\begin{vmatrix}
		\theta_{n,1}&\dots&\theta_{n+N,1}\\
		\vdots&\ddots&\vdots\\
		\theta_{n,N}&\dots&\theta_{n+N,N}\\
		\phi_{n}&\dots&\phi_{n+N}
	\end{vmatrix} 
\end{equation}
where $\theta_{n,k}, k=1,\cdots,N$ and $\phi_n$ satisfy the dispersion relations reduced from Lax pair \eqref{kLP1}$\sim$\eqref{kLP2} 
\begin{align}
	&D_1\phi_n=-\phi_{n+1},\\
	&D_2\phi_n=\phi_{n-1}.
\end{align}
In this sense, we construct Casoratian solutions to the $q$-2DTL equation \eqref{qTL} and its B$\ddot{a}$cklund transformation \eqref{kBT1}$\sim$\eqref{kBT2} by Darboux transformation. These solutions are exactly the same as the ones obtained by Hirota's method when we replace $\phi_n$ by $\theta_{n,N+1}$ in \eqref{DTS}.

\begin{rem}
	We have discarded the factor $(-1)^N$ in the expressions of $X_{n+1}[N+1]$ and $\phi_{n+1}[N+1]$ since it is not essential. 
\end{rem}

\section{Binary Darboux transformation for the q-2DTL equation}
In this section, we will construct binary Darboux transformation for the $q$-2DTL equation \eqref{cbQT} which will be shown to give Grammian solutions to the bilinear $q$-2DTL equation \eqref{qTL} and its B$\ddot{a}$cklund transformation \eqref{kBT1}$\sim$\eqref{kBT2}. These solutions will be proved both by induction and direct verification. In addition, we will present these solutions expressed in terms of quantum integrals. 

\subsection{Grammian solutions to the bilinear $q$-2DTL equation and its B$\ddot{a}$cklund transformation}
The linear equations \eqref{kLP1} and \eqref{kLP2} have the formal adjoints 
\begin{eqnarray}
	&&D_1\psi_n=\sigma_1(\psi_{n-1})-J_n^\dagger\sigma_1(\psi_n),\label{AF1}\\
	&&D_2\psi_n=-V^\dagger_n\sigma_2(\psi_{n+1}),\label{AF2}
\end{eqnarray}
where $\dagger$ presents Hermite conjugate. 

Following the standard procedure of a binary Darboux transformation, we introduce a potential $\Omega(\theta_n,\rho_n)$ satisfying the three conditions:
\begin{eqnarray}
	&&D_1\Omega(\theta_n,\rho_n)=\sigma_1(\rho^\dagger_n)\theta_{n+1},\\
	&&D_2\Omega(\theta_n,\rho_n)=\sigma_2(\rho^\dagger_{n+1})V_n\theta_n,\\
	&&\Omega(\theta_n,\rho_n)-\Omega(\theta_{n-1},\rho_{n-1})=-\rho^\dagger_n\theta_n.
\end{eqnarray}
A binary Darboux transformation is then defined by 
\begin{eqnarray}
	&&\tilde{\phi_n}=\phi_n-\theta_n\Omega(\theta_{n-1},\rho_{n-1})^{-1}\Omega(\phi_{n-1},\rho_{n-1}),\\
	&&\tilde{\psi_n}=\psi_n-\rho_n\Omega(\theta_n,\rho_n)^{-\dagger}\Omega(\theta_n,\psi_n)^\dagger,\\
	&&\tilde X_n=(1-\theta_{n}\Omega(\theta_{n-1},\rho_{n-1})^{-1}\rho_{n}^\dagger)X_n,
\end{eqnarray}
Besides the set of particular eigenfunctions $\theta_{n,i}, i=1,\dots,N,N+1$ of the Lax pair \eqref{kLP1} and \eqref{kLP2}, Let $\rho_{n,j}$ for $j=1,\dots,N$ be a set of particular eigenfunctions of the adjoint Lax pair  \eqref{AF1} and \eqref{AF2}. Then the binary Darboux transformation can be defined recursively by
\begin{eqnarray*}
	&&\phi_n[N+1]=\phi_n[N]-\theta_n[N]\Omega(\theta_{n-1}[N],\rho_{n-1}[N])^{-1}\Omega(\phi_{n-1}[N],\rho_{n-1}[N]),\\
	&&\psi_n[N+1]=\psi_n[N]-\rho_n[N]\Omega(\theta_n[N],\rho_n[N])^{-\dagger}\Omega(\theta_n[N],\psi_n[N])^\dagger,\\
	&&X_n[N+1]=(1-\theta_{n}[N]\Omega(\theta_{n-1}[N],\rho_{n-1}[N])^{-1}\rho_{n}[N]^\dagger)X_n[N],
\end{eqnarray*}
where $\phi_n[1]=\phi_n$, $\psi_n[1]=\psi_n$, $X_n[1]=X_n$ and 
\begin{equation*}
	\theta_n[N]=\phi_n[N]|_{\phi_n\rightarrow \theta_{n,N}},\,\, \rho_n[N]=\psi_n[N]|_{\psi_n\rightarrow \rho_{n,N}}
\end{equation*}

Using the notation $\Theta_n=(\theta_{n,1},\dots,\theta_{n,N})$ and $P_n=(\rho_{n,1},\dots,\rho_{n,N})$, it is easy to prove by induction that for $N\ge 1$, 
\begin{align}
	\phi_n[N+1]&=\begin{vmatrix}
		\Omega(\Theta_{n-1},P_{n-1})&\Omega(\phi_{n-1},P_{n-1})\\
		\Theta_n&\phi_n\end{vmatrix}\cdot\begin{vmatrix}
		\Omega(\Theta_{n-1},P_{n-1})\end{vmatrix}^{-1},\label{NE}\\
	\psi_n[N+1]&=\begin{vmatrix}
		\Omega(\Theta_{n},P_{n})^\dagger&\Omega(\Theta_{n},\psi_{n})^\dagger\\
		P_n&\psi_n
	\end{vmatrix}\cdot\begin{vmatrix}
		\Omega(\Theta_{n},P_{n})^\dagger
	\end{vmatrix}^{-1}\label{NAE}
\end{align}
and
\begin{equation} \Omega(\phi_n[N+1],\psi_n[N+1])=
	\begin{vmatrix}
		\Omega(\Theta_n,P_n)&\Omega(\phi_n,P_n)\\
		\Omega(\Theta_n,\psi_n)&\Omega(\phi_n,\psi_n)
	\end{vmatrix}\cdot\begin{vmatrix}
		\Omega(\Theta_n,P_n)
	\end{vmatrix}^{-1}.\label{NP}
\end{equation}
We may thus after $N$-step binary Darboux transformations obtain
\begin{equation}
	X_n[N+1]=
		\begin{vmatrix}
		\Omega(\Theta_{n-1},P_{n-1})&P_n^\dagger\\
		\Theta_n&1
	\end{vmatrix}\cdot\begin{vmatrix}
		\Omega(\Theta_{n-1},P_{n-1}) \end{vmatrix}^{-1}X_n.\label{QG}
\end{equation}
In fact, we can prove the above results by induction. 
\begin{align*}
	X_n[N+2]=&(1-\theta_n[N+1]\Omega(\Theta_{n-1}[N+1],P_{n-1}[N+1])^{-1}\rho_n[N+1]^\dagger)X_n[N+1]\\
	=&\left(1-\begin{vmatrix}
		\Omega(\Theta_{n-1},P_{n-1})&\Omega(\theta_{n-1,N+1},P_{n-1})\\
		\Theta_n&\theta_{n,N+1}\end{vmatrix}\cdot \begin{vmatrix}
		\Omega(\Theta_{n-1},P_{n-1})\end{vmatrix}^{-1}\right.\\
	&\cdot\begin{vmatrix}
		\Omega(\Theta_{n-1},P_{n-1})
	\end{vmatrix}\cdot \begin{vmatrix}
		\Omega(\Theta_{n-1},P_{n-1})&\Omega(\theta_{n-1,N+1},P_{n-1})\\
		\Omega(\Theta_{n-1},\rho_{n-1,N+1})&\Omega(\theta_{n-1,N+1},\rho_{n-1,N+1})
	\end{vmatrix}^{-1}\\
	&\left.\cdot \begin{vmatrix}
		\Omega(\Theta_{n},P_{n})&P_n^\dagger\\
		\Omega(\Theta_{n},\rho_{n,N+1})&\rho_{n,N+1}^\dagger
	\end{vmatrix}\cdot\begin{vmatrix}
		\Omega(\Theta_{n},P_{n})
	\end{vmatrix}^{-1}\right)\cdot \begin{vmatrix}
		\Omega(\Theta_{n},P_{n})\end{vmatrix}\cdot\begin{vmatrix}
		\Omega(\Theta_{n-1},P_{n-1}) \end{vmatrix}^{-1}X_n\\
		=&\left(\begin{vmatrix}
		\Omega(\Theta_{n},P_{n})\end{vmatrix}-\begin{vmatrix}
		\Omega(\Theta_{n-1},P_{n-1})&\Omega(\theta_{n-1,N+1},P_{n-1})\\
		\Theta_n&\theta_{n,N+1}\end{vmatrix}\right.\\
	&\cdot \begin{vmatrix}
		\Omega(\Theta_{n-1},P_{n-1})&\Omega(\theta_{n-1,N+1},P_{n-1})\\
		\Omega(\Theta_{n-1},\rho_{n-1,N+1})&\Omega(\theta_{n-1,N+1},\rho_{n-1,N+1})
	\end{vmatrix}^{-1}\left.\cdot \begin{vmatrix}
		\Omega(\Theta_{n},P_{n})&P_n^\dagger\\
		\Omega(\Theta_{n},\rho_{n,N+1})&\rho_{n,N+1}^\dagger
	\end{vmatrix}\right)\cdot \begin{vmatrix}
		\Omega(\Theta_{n-1},P_{n-1}) \end{vmatrix}^{-1}X_n
\end{align*}
Noticing
\begin{align*}
\begin{vmatrix}
	\Omega(\Theta_{n},P_{n})\end{vmatrix}=
		\begin{vmatrix}
		\Omega(\Theta_{n-1},P_{n-1})&P_n^\dagger\\
		\Theta_n&1
	\end{vmatrix}
\end{align*}
and
\begin{align*}
	\begin{vmatrix}
		\Omega(\Theta_{n},P_{n})&P_n^\dagger\\
		\Omega(\Theta_{n},\rho_{n,N+1})&\rho_{n,N+1}^\dagger
	\end{vmatrix}=\begin{vmatrix}\Omega(\Theta_{n-1},P_{n-1})&P_n^\dagger\\
		\Omega(\Theta_{n-1},\rho_{n-1,N+1})&\rho_{n,N+1}^\dagger
	\end{vmatrix},
\end{align*}
by using Jacobi identity, 
\begin{align*}
	&\begin{vmatrix}
		\Omega(\Theta_{n-1},P_{n-1})&\Omega(\theta_{n-1,N+1},P_{n-1})&P_n^\dagger\\
		\Omega(\Theta_{n-1},\rho_{n-1,N+1})&\Omega(\theta_{n-1,N+1},\rho_{n-1,N+1})&\rho_{n,N+1}^\dagger\\
		\Theta_n&\theta_{n,N+1}&1
	\end{vmatrix}\cdot \begin{vmatrix}
		\Omega(\Theta_{n-1},P_{n-1}) \end{vmatrix}\\ 
	&=\begin{vmatrix}
		\Omega(\Theta_{n-1},P_{n-1})&P_n^\dagger\\
		\Theta_n&1
	\end{vmatrix}\cdot  \begin{vmatrix}
		\Omega(\Theta_{n-1},P_{n-1})&\Omega(\theta_{n-1,N+1},P_{n-1})\\
		\Omega(\Theta_{n-1},\rho_{n-1,N+1})&\Omega(\theta_{n-1,N+1},\rho_{n-1,N+1})
	\end{vmatrix}\\
	&\quad- \begin{vmatrix}
		\Omega(\Theta_{n-1},P_{n-1})&\Omega(\theta_{n-1,N+1},P_{n-1})\\
		\Theta_n&\theta_{n,N+1}\end{vmatrix} \cdot  \begin{vmatrix}\Omega(\Theta_{n-1},P_{n-1})&P_n^\dagger\\
		\Omega(\Theta_{n-1},\rho_{n-1,N+1})&\rho_{n,N+1}^\dagger  \end{vmatrix}    
		 \end{align*}
we have
\begin{align*}
	X_n[N+2]
	=& \begin{vmatrix}
		\Omega(\Theta_{n-1},P_{n-1})&\Omega(\theta_{n-1,N+1},P_{n-1})&P_n^\dagger\\
		\Omega(\Theta_{n-1},\rho_{n-1,N+1})&\Omega(\theta_{n-1,N+1},\rho_{n-1,N+1})&\rho_{n,N+1}^\dagger\\
		\Theta_n&\theta_{n,N+1}&1
	\end{vmatrix}\\
	&\cdot\begin{vmatrix}
		\Omega(\Theta_{n-1},P_{n-1})&\Omega(\theta_{n-1,N+1},P_{n-1})\\
		\Omega(\Theta_{n-1},\rho_{n-1,N+1})&\Omega(\theta_{n-1,N+1},\rho_{n-1,N+1})
	\end{vmatrix}^{-1}X_n.
\end{align*}

Similarly, because of $X_{n+1}=\tau_{n+1}\tau_{n}^{-1}$ and $\phi_{n+1}=\tau_n'\tau_n^{-1}$, one can easily derive Grammian solutions to the bilinear $q$-2DTL equation \eqref{qTL} and its B$\ddot a$cklund transformation \eqref{kBT1}$\sim$\eqref{kBT2} from the expressions for $\phi_n[N+1]$ and $X_n[N+1]$ given by \eqref{NE} and \eqref{QG}. Take the seed solution $X_n=1$, we have 
\begin{equation}\label{QTBFS}
	\tau_n=
	\begin{vmatrix}
		\Omega(\Theta_{n-1},P_{n-1}) \end{vmatrix}=|\Omega(\theta_{n-1,j},\rho_{n-1,i})|_{1\le i,j\le N},\,\, \tau_n'=\begin{vmatrix}
		\Omega(\Theta_{n-1},P_{n-1})&\Omega(\phi_{n-1},P_{n-1})\\
		\Theta_n&\phi_n\end{vmatrix}.\end{equation}
Correspondingly, $\theta_{n,i}$ and $\rho_{n,i}$ satisfy dispersion relations produced from the Lax pair \eqref{kLP1}$\sim$\eqref{kLP2} and the adjoint Lax pair \eqref{AF1}$\sim$\eqref{AF2}
\begin{align}
	D_1{\theta_{n,i}}&=-\theta_{n+1,i},\, D_2\theta_{n,i}=\theta_{n-1,i},\label{qDR1}\\
D_1\rho_{n,i}&=\sigma_1(\rho_{n-1,i}),\, D_2\rho_{n,i}=-\sigma_2(\rho_{n+1,i}),\label{qDR2}
\end{align}
and $\Omega(\theta_{n-1,j},\rho_{n-1,i})$ is defined by
\begin{align}
	&D_1\Omega(\theta_{n,j},\rho_{n,i})=\sigma_1(\rho_{n,i}^\dagger)\theta_{n+1,j},\label{EX1}\\
	&D_2\Omega(\theta_{n,j},\rho_{n,i})=\sigma_2(\rho_{n+1,i}^\dagger)\theta_{n,j},\label{EX2}\\
	&\Omega(\theta_{n,j},\rho_{n,i})-\Omega(\theta_{n-1,j},\rho_{n-1,i})=-\rho_{n,i}^\dagger\theta_{n,j}. \label{EX3}
\end{align}

In this sense, we construct Grammian solutions to the bilinear $q$-2DTL equation \eqref{qTL} and its B$\ddot a$cklund transformation \eqref{kBT1}$\sim$\eqref{kBT2} by binary Darboux transformation.

\subsection{Direct verifications}
In this part, we will provide direct verifications of Grammian solutions given by \eqref{QTBFS} to the bilinear $q$-2DTL equation \eqref{qTL} and its B$\ddot{a}$cklund transformation \eqref{kBT1}$\sim$\eqref{kBT2}, respectively. 

Actually, we have after detailed calculations 
\begin{align}
	&\sigma_2(\tau_{n+1})=\begin{vmatrix}\Omega(\Theta_{n-1},P_{n-1})&\sigma_2(P_n^\dagger)\\
		\Theta_n&1
	\end{vmatrix},\label{CS1}\\
	&\sigma_1(\tau_n)=-a(x)\begin{vmatrix}\Omega(\Theta_{n-1},P_{n-1})&\sigma_1(P_{n-1}^\dagger)\\
		\Theta_n&-a(x)^{-1}
	\end{vmatrix}\label{CS2},\\
	&\sigma_1(\tau_{n-1})=-\begin{vmatrix}
		\Omega(\Theta_{n-1},P_{n-1})&\sigma_1(P_{n-1}^\dagger)\\
		\Theta_{n-1}&-1
	\end{vmatrix},\label{CS3}\\
	&(1+a(x)b(y))\sigma_1\sigma_2(\tau_n)=a(x)b(y)\begin{vmatrix}\Omega(\Theta_{n-1},P_{n-1})&\sigma_1(P_{n-1}^\dagger)&\sigma_2(P_n^\dagger)\\
		\Theta_n&-a(x)^{-1}&1\\
		\Theta_{n-1}&-1&-b(y)^{-1}
	\end{vmatrix}.\label{CS4}
\end{align}
Substituting \eqref{CS1}$\sim$\eqref{CS4} into \eqref{qTL}, we get the Jacobi identity
\begin{align*}
	&\begin{vmatrix}\Omega(\Theta_{n-1},P_{n-1})&\sigma_1(P_{n-1}^\dagger)&\sigma_2(P_n^\dagger)\\
		\Theta_n&a(x)^{-1}&1\\
		\Theta_{n-1}&-1&-b(y)^{-1}
	\end{vmatrix}\begin{vmatrix}\Omega(\Theta_{n-1},P_{n-1})
	\end{vmatrix}\\
	&=\begin{vmatrix}\Omega(\Theta_{n-1},P_{n-1})&\sigma_1(P_{n-1}^\dagger)\\
		\Theta_n&-a(x)^{-1}
	\end{vmatrix}\begin{vmatrix}\Omega(\Theta_{n-1},P_{n-1})&\sigma_2(P_n^\dagger)\\
		\Theta_{n-1}&-b(y)^{-1}
	\end{vmatrix}\\
	&\quad-\begin{vmatrix}\Omega(\Theta_{n-1},P_{n-1})&\sigma_2(P_n^\dagger)\\
		\Theta_n&1\\
	\end{vmatrix}\begin{vmatrix}\Omega(\Theta_{n-1},P_{n-1})&\sigma_1(P_{n-1}^\dagger)\\
		\Theta_{n-1}&-1
	\end{vmatrix}.
\end{align*}
Thus, $\tau_{n}$ given by \eqref{QTBFS} is a solution to the bilinear $q$-2DTL equation \eqref{qTL}.

Similarly, we can calculate to obtain
\begin{align}
	&\tau_{n-1}'=\begin{vmatrix}\Omega(\Theta_{n-1},P_{n-1})&\Omega(\phi_{n-1},P_{n-1})\\
		\Theta_{n-1}&\phi_{n-1}
	\end{vmatrix},\label{CS5}\\
	&\sigma_2(\tau_{n}')=-b(y)\begin{vmatrix}\Omega(\Theta_{n-1},P_{n-1})&\Omega(\phi_{n-1},P_{n-1})&\sigma_2(P_n^{\dagger})\\
		\Theta_{n}&\phi_{n}&1\\
		\Theta_{n-1}&\phi_{n-1}&-b(y)^{-1}
	\end{vmatrix},\label{CS6}\\
	&\sigma_1(\tau_{n-1}')=-a(x)\begin{vmatrix}\Omega(\Theta_{n-1},P_{n-1})&\Omega(\phi_{n-1},P_{n-1})&\sigma_1(P_{n-1}^{\dagger})\\
		\Theta_{n-1}&\phi_{n-1}&-1\\
		\Theta_{n}&\phi_{n}&-a(x)^{-1}
	\end{vmatrix}.\label{CS7}
\end{align}
Substituting \eqref{CS1}$\sim$\eqref{CS7} into \eqref{kBT1} and \eqref{kBT2} will lead to the following two Jacobi identities
\begin{align*}
	&\begin{vmatrix}\Omega(\Theta_{n-1},P_{n-1})&\Omega(\phi_{n-1},P_{n-1})&\sigma_1(P_{n-1}^{\dagger})\\
		\Theta_{n-1}&\phi_{n-1}&-1\\
		\Theta_{n}&\phi_{n}&-a(x)^{-1}
	\end{vmatrix}\begin{vmatrix}\Omega(\Theta_{n-1},P_{n-1})\end{vmatrix})\\
	&=\begin{vmatrix}\Omega(\Theta_{n-1},P_{n-1})&\Omega(\phi_{n-1},P_{n-1})\\
		\Theta_{n-1}&\phi_{n-1}
	\end{vmatrix}\begin{vmatrix}\Omega(\Theta_{n-1},P_{n-1})&\sigma_1(P_{n-1}^{\dagger})\\
		\Theta_{n}&-a(x)^{-1}
	\end{vmatrix}\\
	&-\begin{vmatrix}\Omega(\Theta_{n-1},P_{n-1})&\Omega(\phi_{n-1},P_{n-1})\\
		\Theta_{n}&\phi_{n}
	\end{vmatrix}\begin{vmatrix}\Omega(\Theta_{n-1},P_{n-1})&\sigma_1(P_{n-1}^{\dagger})\\
		\Theta_{n-1}&-1
	\end{vmatrix}
\end{align*}
and
\begin{align*}
	&\begin{vmatrix}\Omega(\Theta_{n-1},P_{n-1})&\Omega(\phi_{n-1},P_{n-1})&\sigma_2(P_n^{\dagger})\\
		\Theta_{n}&\phi_{n}&1\\
		\Theta_{n-1}&\phi_{n-1}&-b(y)^{-1}
	\end{vmatrix}\begin{vmatrix}\Omega(\Theta_{n-1},P_{n-1})\end{vmatrix})\\
	&=\begin{vmatrix}\Omega(\Theta_{n-1},P_{n-1})&\Omega(\phi_{n-1},P_{n-1})\\
		\Theta_{n}&\phi_{n}
	\end{vmatrix}\begin{vmatrix}\Omega(\Theta_{n-1},P_{n-1})&\sigma_2(P_n^{\dagger})\\
		\Theta_{n-1}&b(y)^{-1}
	\end{vmatrix}\\
	&-\begin{vmatrix}\Omega(\Theta_{n-1},P_{n-1})&\Omega(\phi_{n-1},P_{n-1})\\
		\Theta_{n-1}&\phi_{n-1}
	\end{vmatrix}\begin{vmatrix}\Omega(\Theta_{n-1},P_{n-1})&\sigma_2(P_n^{\dagger)}\\
		\Theta_{n}&1
	\end{vmatrix}.
\end{align*}
This indicates that $\tau_n$ and $\tau_n'$ given by \eqref{QTBFS} are indeed Grammian solutions to the bilinear B$\ddot{a}$cklund transformation \eqref{kBT1}$\sim$\eqref{kBT2}. 
%

\subsection{Quantum integral representation}
%
%

We will first go over some results in quantum calculus \cite{VKPC}. Denote the $q$-shift operator $\hat{M}_qF(x)=F(qx)$. The $q$-difference operator is defined as
\begin{equation*}
	D_{q,x}F(x)=\frac{F(qx)-F(x)}{(q-1)x}=\frac{1}{(q-1)x}(\hat{M}_q-	1)F(x).
\end{equation*}
Given a function $f(x)$, its anti-derivative $F(x)$ for $D_{q,x}F(x)=f(x)$ can be written as
\begin{equation*}
	F(x)=\frac{1}{1-\hat{M}_q}(1-q)xf(x)=(1-q)x\sum_{j=0}^{\infty}q^jf(q^jx).
\end{equation*}
In other words, the anti-derivative of $f(x)$ or the quantum integral of $f(x)$ is 
\begin{equation*}
	\int f(x)d_{q}x=F(x)=(1-q)x\sum_{j=0}^{\infty}q^jf(q^jx).
\end{equation*}

The $q$-difference operator adopted in this paper is a bit different. However, we can still define quantum integrals similarly.
For the $q$-difference operator defined by 
\begin{align*}
D_1F(x)=\frac{1}{(q-1)x}(\sigma_1-1)F(x),
\end{align*}
assume that $D_1F(x)=f(x)$ for a given function $f(x)$, then we have
\begin{align*}
F(x)=\frac{1}{1-\sigma_1}(1-q)xf(x)=(1-q)x\sum_{j=0}^{\infty}q^{\alpha j}f(q^{\alpha j}x).
\end{align*}
Therefore, we can write the quantum integral of $f(x)$ formally as
\begin{align*}
\int f(x)d_{q^\alpha }x=F(x)=(1-q)x\sum_{j=0}^{\infty}q^{\alpha j}f(q^{\alpha j}x).
\end{align*}

Since
\begin{align*}
 D_1\Omega(\theta_{n-1,j},\rho_{n-1,i})=\sigma_1(\rho_{n-1,i}^\dagger)\theta_{n,j},
\end{align*}
we have 
\begin{align*}
\Omega(\theta_{n-1,j},\rho_{n-1,i})=\int \sigma_1(\rho_{n-1,i}^\dagger)\theta_{n,j} d_{q^\alpha }x
\end{align*}
which makes it possible to express the Grammian solution to the bilinear $q$-2DTL equation \eqref{qTL} in terms of quantum integals
\begin{equation}\label{GQI}
	\tau_n=
	\begin{vmatrix}
		\Omega(\Theta_{n-1},P_{n-1}) \end{vmatrix}=\begin{vmatrix}
		\Omega(\theta_{n-1,j},\rho_{n-1,i}) \end{vmatrix}_{1\le i,j\le N}=
	\begin{vmatrix}\int \sigma_1(\rho_{n-1,i}^\dagger)\theta_{n,j}d_{q^\alpha}x \end{vmatrix}.
\end{equation}
In the same way, one can express $\tau_n'$ in terms of quantum integrals. 
\begin{rem}
According to the definition of quantum integrals, the Grammian solution \eqref{GQI} can be expressed in terms of formal series, too. A similar Gramm-type determinant solution for the bilinear $q$-2DTL equation \eqref{qTL} was reported in \cite{WHT}. We would like to point out that these two Grammian solutions are equivalent under certain transformations. 
\end{rem}

\section{Periodic reductions}
It is well known that periodic 2 reduction and periodic 3 reduction of the two-dimensional Toda lattice equation produce sine-Gordon equation and Tzitzeica equation, individually. It is of great research interest to study periodic reductions of $q$-2DTL equations as well. In this section, we will consider the $2$-periodic reduction of the $q$-2DTL equation \eqref{cbQT}. We propose a $q$-difference sine-Gordon ($q$-sG) equation, a modified $q$-difference sine-Gordon ($q$-msG) equation and present their solutions for the first time. 

\subsection{The $q$-difference sine-Gordon equation and solutions}
By imposing the periodic condition $X_n=X_{n+2}$ on the $q$-2DTL equation \eqref{cbQT}, we have  
\begin{align}
	D_2(D_1(X_0)X_0^{-1})&=\sigma_2(X_{1})X_0^{-1}-\sigma_1(\sigma_2(X_0)X_{1}^{-1}),\label{qsG1}\\
	D_2(D_1(X_1)X_1^{-1})&=\sigma_2(X_{0})X_1^{-1}-\sigma_1(\sigma_2(X_1)X_{0}^{-1}).\label{qsG2}
\end{align}
If we assume $X=X_0=\bar{X_1}$ where $\bar{X_1}$ represents the conjugate of $X_1$, then \eqref{qsG1} and \eqref{qsG2} are reduced to a single equation
\begin{align}
D_2(D_1(X)X^{-1})&=\sigma_2(\bar{X})X^{-1}-\sigma_1(\sigma_2(X)\bar{X}^{-1}). \label{rqsG}
\end{align}

 Rewrite $X=\lambda_1 f/\bar{f}$ with $\lambda_1$ being a real constant.  By taking the continuum limit $q\rightarrow 1$, \eqref{rqsG} will transform to   
\begin{align*}
	(\ln f/\bar{f})_{xy}&=\bar{f}^2/f^2-1+1-f^2/\bar{f}^2
\end{align*}
which implies 
\begin{align*}
	(\ln f)_{xy}&=\bar{f}^2/f^2-1
\end{align*}
or equivalently,
\begin{align}
D_xD_y f\cdot f&=2\bar {f}^2-2f^2.\label{bsG}
\end{align}
It is well known that \eqref{bsG} is nothing but the bilinear sine-Gordon equation. In this sense, we call \eqref{qsG1} and \eqref{qsG2} the $q$-sG equation, and \eqref{rqsG} the reduced $q$-sG equation.

Before we proceed to derive determinant solutions to the $q$-sG equation, we need to review some properties of $q$-exponential functions. There are two kinds of $q$-exponential functions. But the $q$-exponential functions adopted here are slightly altered from the original ones in \cite{VKPC}. Given $q$-exponential functions
\begin{align*}
	&e_{q^\alpha}^x=\sum\limits_{j=0}^\infty \frac{x^j}{[j]_{q^\alpha}!},\ \,e_{q\beta}^y=\sum\limits_{k=0}^\infty \frac{y^k}{[k]_{q^\beta}!},\\
	&E_{q^\alpha}^x=\sum\limits_{j=0}^\infty ({q^\alpha})^{j(j-1)/2} \frac{x^j}{[j]_{q^\alpha}!},\,\  E_{q^\beta}^y=\sum\limits_{k=0}^\infty ({q^\beta})^{k(k-1)/2} \frac{y^k}{[k]_{q^\beta}!}
\end{align*}
with $[j]_{q^\alpha}!=[j]_{q^\alpha}[j-1]_{q^\alpha}\cdots [1]_{q^\alpha},\,[j]_{q^\alpha}=\frac{(q^\alpha)^j-1}{q-1},\,[k]_{q^\beta}!=[k]_{q^\beta}[k-1]_{q^\beta}\cdots [1]_{q^\beta}$ and $[k]_{q^\beta}=\frac{(q^\beta)^j-1}{q-1}$.
Noticing $q^{j(j-1)/2}=q^{(j-1)(j-2)/2+(j-1)}$, it is not difficult to prove that 
\begin{align*}
	&D_1e_{q^\alpha}^x=e_{q^\alpha}^x,\,\ D_1e_{q^\alpha}^{px}=pe_{q^\alpha}^{px},\\
	&D_1 E_{q^\alpha}^x=E_{q^\alpha}^{q^\alpha x},\,\ D_1 E_{q^\alpha}^{px}=pE_{q^\alpha}^{q^\alpha( px)}.
\end{align*}

As is shown before, the $q$-2DTL equation \eqref{cbQT} has solutions $X_{n+1}=\tau_{n+1}\tau_{n}^{-1}$ with $\tau_n$ given by \eqref{QTBFS}. In what follows, we will construct solutions to \eqref{qsG1} and \eqref{qsG2} by making $2$-periodic reduction on $X_{n+1}$. First we choose the simplest non-trivial solutions of \eqref{qDR1} and \eqref{qDR2}
\begin{align*}
	\theta_{n,j}=b_j(p_j)^{n}e_{q^\alpha}^{-p_j x}e_{q^\beta}^{\frac{1}{p_j}y},\,\rho_{n,i}=a_i(r_i)^{-n}E_{q^\alpha}^{r_i x}E_{q^\beta}^{-\frac{1}{r_i}y}
\end{align*}
where $a_i$ and $b_j$ are constants and then we obtain 
\begin{align*}
	D_1(\rho_{n,i}\theta_{n+1,j})&=D_1(\rho_{n,i})\theta_{n+1,j}+\sigma_1(\rho_{n,i})D_1(\theta_{n+1,j})\\
	&=(r_i-p_j)\sigma_1(\rho_{n,i})\theta_{n+1,j}\\
	&=(r_i-p_j)\Omega(\theta_{n,j},\rho_{n,i})
\end{align*}
which implies
\begin{align*}
	\Omega(\theta_{n,j},\rho_{n,i})&=\delta_{i,j}+\frac{1}{r_i-p_j}\rho_{n,i}\theta_{n+1,j} \\
%
	&=\delta_{i,j}+\frac{a_ib_jp_j}{r_i-p_j}\left(\frac{p_j}{r_i}\right)^ne_{q^\alpha}^{-p_j x}e_{q^\beta}^{\frac{1}{p_j}y}E_{q^\alpha}^{r_i x}E_{q^\beta}^{-\frac{1}{r_i}y}\\
	&=\left(\frac{p_j}{r_i}\right)^n\left( \delta_{i,j}\left(\frac{r_i}{p_j}\right)^n  + \frac{a_ib_jp_j}{r_i-p_j}e_{q^\alpha}^{-p_j x}e_{q^\beta}^{\frac{1}{p_j}y}E_{q^\alpha}^{r_i x}E_{q^\beta}^{-\frac{1}{r_i}y}  \right).
\end{align*}
The choice of constant of integration as $\delta_{i,j}$ is needed to effect the periodic reduction that will be made shortly. 

With these prerequisites, we have
\begin{align*}
	X_{n+1}&=\tau_{n+1}\tau_n^{-1}=\begin{vmatrix}
		\Omega(\Theta_{n-1},P_{n-1})&P_n^\dagger\\
		\Theta_n&1
	\end{vmatrix}\cdot\begin{vmatrix}
		\Omega(\Theta_{n-1},P_{n-1}) \end{vmatrix}^{-1}\\
	&=\frac{\begin{vmatrix}
			\left(\frac{p_j}{r_i}\right)^{n-1}\left( \delta_{i,j}\left(\frac{r_i}{p_j}\right)^{n-1}  + \frac{a_ib_jp_j}{r_i-p_j}e_{q^\alpha}^{-p_j x}e_{q^\beta}^{\frac{1}{p_j}y}E_{q^\alpha}^{r_i x}E_{q^\beta}^{-\frac{1}{r_i}y}  \right)& \left(a_i(r_i)^{-n}E_{q^\alpha}^{r_i x}E_{q^\beta}^{-\frac{1}{r_i}y}\right)^\dagger\\
			b_j(p_j)^{n}e_{q^\alpha}^{-p_j x}e_{q^\beta}^{\frac{1}{p_j}y}&1
	\end{vmatrix}}{\begin{vmatrix}
			\left(\frac{p_j}{r_i}\right)^{n-1}\left( \delta_{i,j}\left(\frac{r_i}{p_j}\right)^{n-1}  + \frac{a_ib_jp_j}{r_i-p_j}e_{q^\alpha}^{-p_j x}e_{q^\beta}^{\frac{1}{p_j}y}E_{q^\alpha}^{r_i x}E_{q^\beta}^{-\frac{1}{r_i}y}  \right) \end{vmatrix}}\\
	&=\frac{\begin{vmatrix}
			\left(\frac{p_j}{r_i}\right)^{n}\left( \delta_{i,j}\left(\frac{r_i}{p_j}\right)^{n}  + \frac{a_ib_jr_i}{r_i-p_j}e_{q^\alpha}^{-p_j x}e_{q^\beta}^{\frac{1}{p_j}y}E_{q^\alpha}^{r_i x}E_{q^\beta}^{-\frac{1}{r_i}y}  \right)& \left(a_i(r_i)^{-n}E_{q^\alpha}^{r_i x}E_{q^\beta}^{-\frac{1}{r_i}y}\right)^\dagger\\
			b_j(p_j)^{n}e_{q^\alpha}^{-p_j x}e_{q^\beta}^{\frac{1}{p_j}y}&1
	\end{vmatrix}}{\begin{vmatrix}
			\left(\frac{p_j}{r_i}\right)^{n}\left( \delta_{i,j}\left(\frac{r_i}{p_j}\right)^{n}  + \frac{a_ib_jr_i}{r_i-p_j}e_{q^\alpha}^{-p_j x}e_{q^\beta}^{\frac{1}{p_j}y}E_{q^\alpha}^{r_i x}E_{q^\beta}^{-\frac{1}{r_i}y}  \right) \end{vmatrix}}\\
	&=\frac{\begin{vmatrix}
			\delta_{i,j}\left(\frac{r_i}{p_j}\right)^{n}  + \frac{a_ib_jr_i}{r_i-p_j}e_{q^\alpha}^{-p_j x}e_{q^\beta}^{\frac{1}{p_j}y}E_{q^\alpha}^{r_i x}E_{q^\beta}^{-\frac{1}{r_i}y}  & \left(a_iE_{q^\alpha}^{r_i x}E_{q^\beta}^{-\frac{1}{r_i}y}\right)^\dagger\\
			b_je_{q^\alpha}^{-p_j x}e_{q^\beta}^{\frac{1}{p_j}y}&1
	\end{vmatrix}}{\begin{vmatrix}
			\delta_{i,j}\left(\frac{r_i}{p_j}\right)^{n}  + \frac{a_ib_jr_i}{r_i-p_j}e_{q^\alpha}^{-p_j x}e_{q^\beta}^{\frac{1}{p_j}y}E_{q^\alpha}^{r_i x}E_{q^\beta}^{-\frac{1}{r_i}y}   \end{vmatrix}}.
\end{align*}
It is obvious from this expression for $X_{n+1}$ that it is $2$-periodic only if $(r_1/p_1)^2=\cdots=(r_N/p_N)^2=1$, i.e. $r_i=-p_i=\lambda_i$ for $i=1,\dots,N$. Therefore, the $q$-sG equation \eqref{qsG1} and \eqref{qsG2} has the solutions 
\begin{align}
	X_0&=\frac{\begin{vmatrix}
			- \delta_{i,j} + \frac{a_ib_j\lambda_i}{\lambda_i+\lambda_j}e_{q^\alpha}^{\lambda_j x}e_{q^\beta}^{-\frac{1}{\lambda_j}y}E_{q^\alpha}^{\lambda_i x}E_{q^\beta}^{-\frac{1}{\lambda_i}y}  & \left(a_iE_{q^\alpha}^{\lambda_i x}E_{q^\beta}^{-\frac{1}{\lambda_i}y}\right)^\dagger\\
			b_je_{q^\alpha}^{\lambda_j x}e_{q^\beta}^{-\frac{1}{\lambda_j}y}&1
	\end{vmatrix}}{\begin{vmatrix}
			-\delta_{i,j}+ \frac{a_ib_j\lambda_i}{\lambda_i+\lambda_j}e_{q^\alpha}^{\lambda_j x}e_{q^\beta}^{-\frac{1}{\lambda_j}y}E_{q^\alpha}^{\lambda_i x}E_{q^\beta}^{-\frac{1}{\lambda_i}y}   \end{vmatrix}},\label{oS1}\\
	X_1&=\frac{\begin{vmatrix}
			\delta_{i,j}  + \frac{a_ib_j\lambda_i}{\lambda_i+\lambda_j}e_{q^\alpha}^{\lambda_j x}e_{q^\beta}^{-\frac{1}{\lambda_j}y}E_{q^\alpha}^{\lambda_i x}E_{q^\beta}^{-\frac{1}{\lambda_i}y}  & \left(a_iE_{q^\alpha}^{\lambda_i x}E_{q^\beta}^{-\frac{1}{\lambda_i}y}\right)^\dagger\\
			b_je_{q^\alpha}^{\lambda_j x}e_{q^\beta}^{-\frac{1}{\lambda_j}y}&1
	\end{vmatrix}}{\begin{vmatrix}
			\delta_{i,j} + \frac{a_ib_j\lambda_i}{\lambda_i+\lambda_j}e_{q^\alpha}^{\lambda_j x}e_{q^\beta}^{-\frac{1}{\lambda_j}y}E_{q^\alpha}^{\lambda_i x}E_{q^\beta}^{-\frac{1}{\lambda_i}y}  \end{vmatrix}}.	\label{oS2}
			\end{align}
If we further require that $a_i$ are real and $b_j$ are pure imaginary, then we can infer that $X=X_0=\bar{X_1}$, which provides solutions to the reduced $q$-sG equation \eqref{rqsG}. 

\subsection{The modified $q$-difference sine-Gordon equation and solutions}
In Section 2, the generalized bilinear B$\ddot{a}$cklund transformation was proposed for the bilinear $q$-2DTL equation \eqref{qTL}. In the following, we are going to derive B$\ddot{a}$cklund transformation for the $q$-2DTL equation \eqref{cbQT} which paves the way to derive B$\ddot{a}$cklund transformation for the $q$-sG equation \eqref{qsG1} and \eqref{qsG2}.  

Assume that $X_{n+1}'=\tau_{n+1}'\tau_{n}'^{-1}$ and $X_{n+1}=\tau_{n+1}\tau_{n}^{-1}$ are two solutions of \eqref{cbQT}. Notice $\phi_{n+1}=\tau_n'\tau_n^{-1}. $By making use of the Lax pair \eqref{kLP1}$\sim$\eqref{kLP2} and considering the $q$-derivatives of $\phi_{n+1}/\phi_n=X_n'X_n^{-1}$, we have 
\begin{align}
	D_1(\phi_{n+1}/\phi_n)&=D_1(\phi_{n+1})/\phi_n-\sigma_1(\phi_{n+1})D_1(\phi_n)/(\phi_n\sigma_1(\phi_n))\nonumber\\
	&=-\frac{\phi_{n+2}}{\phi_{n+1}} \frac{\phi_{n+1}}{\phi_n}+J_{n+1}\frac{\phi_{n+1}}{\phi_n}-\sigma_1\left(\frac{\phi_{n+1}}{\phi_n}\right)\left(J_n-\frac{\phi_{n+1}}{\phi_n}\right),\label{mLP1}\\
	D_2(\phi_{n+1}/\phi_n)&=D_2(\phi_{n+1})/\phi_n-\sigma_2(\phi_{n+1})D_2(\phi_n)/(\phi_n\sigma_2(\phi_n))\nonumber\\
	&=V_n-V_{n-1}\sigma_2\left(\frac{\phi_{n+1}}{\phi_n}\right)\frac{\phi_{n-1}}{\phi_n},\label{mLP2}
\end{align}
where we have used the equality
\begin{align*}
D_i(f/g)=(D_i(f)g-D_i(g)f)/(g\sigma_i(g)),\,\, i=1,2.
\end{align*}
Apparently, \eqref{mLP1} and \eqref{mLP2} give rise to the B$\ddot a$cklund transformation for \eqref{cbQT}: 
\begin{align}
	D_1(X_n'X_n^{-1})&=-X_{n+1}'X_{n+1}^{-1}X_n'X_n^{-1}+D_1(X_{n+1})X_{n+1}^{-1}X_n'X_n^{-1}\nonumber\\
	&\quad-\sigma_1(X_n'X_n^{-1})(D_1(X_{n})X_{n}^{-1}-X_n'X_n^{-1}),\label{BTsG1}\\
	D_2(X_n'X_n^{-1})&=\sigma_2(X_{n+1})X_n^{-1}-\sigma_2(X_{n})\sigma_2(X_n'X_n^{-1})X_{n-1}'^{-1}.\label{BTsG2}
\end{align}


 By considering the $2$-periodic reduction of \eqref{BTsG1} and \eqref{BTsG2}, we finally obtain the B$\ddot a$cklund transformation for the $q$-sG equation \eqref{qsG1} and \eqref{qsG2}
\begin{align}
D_1(X_0'X_0^{-1})&=-X_{1}'X_{1}^{-1}X_0'X_0^{-1}+D_1(X_{1})X_{1}^{-1}X_0'X_0^{-1}\nonumber\\
	&\quad-\sigma_1(X_0'X_0^{-1})(D_1(X_{0})X_{0}^{-1}-X_0'X_0^{-1}),\label{BTsG3}\\
	D_2(X_0'X_0^{-1})&=\sigma_2(X_{1})X_0^{-1}-\sigma_2(X_{0})\sigma_2(X_0'X_0^{-1})X_{1}'^{-1}.\label{BTsG4}\\
	D_1(X_1'X_1^{-1})&=-X_{0}'X_{0}^{-1}X_1'X_1^{-1}+D_1(X_{0})X_{0}^{-1}X_1'X_1^{-1}\nonumber\\
	&\quad-\sigma_1(X_1'X_1^{-1})(D_1(X_{1})X_{1}^{-1}-X_1'X_1^{-1}),\label{BTs5}\\
	D_2(X_1'X_1^{-1})&=\sigma_2(X_{0})X_1^{-1}-\sigma_2(X_{1})\sigma_2(X_1'X_1^{-1})X_{0}'^{-1}.\label{BTsG6}
\end{align}
If we go one step further, by setting $X=X_0=\bar{X_1}$ and $X'=X_0'=\bar{X_1'}$, we have
\begin{align}
	D_1(X'X^{-1})&=-\bar{X'}\bar{X}^{-1}X'X^{-1}+D_1(\bar{X})\bar{X}^{-1}X'X^{-1}\nonumber\\
	&\quad-\sigma_1(X'X^{-1})(D_1(X)X^{-1}-X'X^{-1}),\label{rBT1}\\
	D_2(X'X^{-1})&=\sigma_2(\bar{X})X^{-1}-\sigma_2(X)\sigma_2(X'X^{-1})\bar{X'}^{-1},\label{rBT2}
\end{align}
which is the B$\ddot{a}$cklund transformation for the reduced $q$-sG equation \eqref{rqsG}.

In fact, \eqref{rBT1} and \eqref{rBT2} transform to
\begin{align*}
	&(\ln X')_x=-\bar{X'}\bar{X}^{-1}+(\ln\bar{X})_x+X'X^{-1},\\
	&(X'X^{-1})_y=\bar{X}X^{-1}-X'\bar{X'}^{-1}
\end{align*}
by taking the continuum limit $q\rightarrow 1$. If we rewrite $X=\lambda_1 f/\bar{f}$ and $X'=\lambda_2 g/\bar{g}$ with $\lambda_i$ for $i=1,2$ being real constants, we have 
\begin{align*}
	&D_x(f \cdot \bar{g})/(f\bar{g})-D_x(\bar{f} \cdot g)/(\bar{f} g) =\lambda_2\lambda_1^{-1}(g\bar{f}/(f\bar{g})-\bar{g}f/(\bar{f}g)),\\
	&fgD_y(\bar{f}\cdot \bar{g})-\bar{f}\bar{g}D_y(f\cdot g)=\lambda_1\lambda_2^{-1}\left((\bar{f}\bar{g})^2-(fg)^2\right)
\end{align*}
which produces the bilinear B$\ddot{a}$cklund transformation for the bilinear sine-Gordon equation \eqref{bsG}
\begin{align*}
	D_x(f\cdot \bar g)&=\lambda_2\lambda_1^{-1}\bar{f}g,\\
	D_y(\bar f\cdot \bar{g})&=-\lambda_1\lambda_2^{-1}fg.
\end{align*}
In this sense, we call \eqref{BTsG3}$\sim$\eqref{BTsG6} the modified $q$-sG equation, and \eqref{rBT1}$\sim$\eqref{rBT2} the reduced modified $q$-sG equation.

It is obvious that \eqref{BTsG1} and \eqref{BTsG2} have solutions $X_{n+1}'=\tau_{n+1}'\tau_{n}'^{-1}$ and $X_{n+1}=\tau_{n+1}\tau_{n}^{-1}$ with $\tau_n$ and $\tau_n'$ given by \eqref{QTBFS}. Denote $\tau_n'\triangleq\tau_n'|_{\phi_n\rightarrow \theta_{n,N+1}}$. By virtue of invariance under elementary row and column operations of determinants, we have
\begin{align*}
X_{n+1}'&=\tau_{n+1}'\tau_{n}'^{-1}\\
&=\begin{vmatrix}
		\Omega(\Theta_{n},P_{n})&\Omega(\phi_{n},P_{n})\\
		\Theta_{n+1}&\phi_{n+1}\end{vmatrix}
\begin{vmatrix}
		\Omega(\Theta_{n-1},P_{n-1})&\Omega(\phi_{n-1},P_{n-1})\\
		\Theta_n&\phi_n\end{vmatrix}^{-1}\\
		&=\frac{\begin{vmatrix}
			\delta_{i,j}+\frac{a_ib_jp_j}{r_i-p_j}\left(\frac{p_j}{r_i}\right)^ne_{q^\alpha}^{-p_j x}e_{q^\beta}^{\frac{1}{p_j}y}E_{q^\alpha}^{r_i x}E_{q^\beta}^{-\frac{1}{r_i}y}&\delta_{i,N+1}+\frac{a_ib_{N+1}p_{N+1}}{r_i-p_{N+1}}\left(\frac{p_{N+1}}{r_i}\right)^ne_{q^\alpha}^{-p_{N+1} x}e_{q^\beta}^{\frac{1}{p_{N+1}}y}E_{q^\alpha}^{r_i x}E_{q^\beta}^{-\frac{1}{r_i}y}\\
			b_j(p_j)^{n+1}e_{q^\alpha}^{-p_j x}e_{q^\beta}^{\frac{1}{p_j}y}&b_{N+1}(p_{N+1})^{n+1}e_{q^\alpha}^{-p_{N+1} x}e_{q^\beta}^{\frac{1}{p_{N+1}}y}
	\end{vmatrix}}{\begin{vmatrix}
			\delta_{i,j}+\frac{a_ib_jp_j}{r_i-p_j}\left(\frac{p_j}{r_i}\right)^{n-1}e_{q^\alpha}^{-p_j x}e_{q^\beta}^{\frac{1}{p_j}y}E_{q^\alpha}^{r_i x}E_{q^\beta}^{-\frac{1}{r_i}y}&\delta_{i,N+1}+\frac{a_ib_{N+1}p_{N+1}}{r_i-p_{N+1}}\left(\frac{p_{N+1}}{r_i}\right)^{n-1}e_{q^\alpha}^{-p_{N+1} x}e_{q^\beta}^{\frac{1}{p_{N+1}}y}E_{q^\alpha}^{r_i x}E_{q^\beta}^{-\frac{1}{r_i}y}\\
			b_j(p_j)^{n}e_{q^\alpha}^{-p_j x}e_{q^\beta}^{\frac{1}{p_j}y}&b_{N+1}(p_{N+1})^{n}e_{q^\alpha}^{-p_{N+1} x}e_{q^\beta}^{\frac{1}{p_{N+1}}y}
	\end{vmatrix}}\\
		&=\frac{\begin{vmatrix}
			\delta_{i,j}(\frac{r_i}{p_j})^n+\frac{a_ib_jp_j}{r_i-p_j}e_{q^\alpha}^{-p_j x}e_{q^\beta}^{\frac{1}{p_j}y}E_{q^\alpha}^{r_i x}E_{q^\beta}^{-\frac{1}{r_i}y}&\delta_{i,N+1}\left(\frac{r_i}{p_{N+1}}\right)^n+\frac{a_ib_{N+1}p_{N+1}}{r_i-p_{N+1}}e_{q^\alpha}^{-p_{N+1} x}e_{q^\beta}^{\frac{1}{p_{N+1}}y}E_{q^\alpha}^{r_i x}E_{q^\beta}^{-\frac{1}{r_i}y}\\
			b_jp_je_{q^\alpha}^{-p_j x}e_{q^\beta}^{\frac{1}{p_j}y}&b_{N+1}p_{N+1}e_{q^\alpha}^{-p_{N+1} x}e_{q^\beta}^{\frac{1}{p_{N+1}}y}
	\end{vmatrix}}{\begin{vmatrix}
			\delta_{i,j}\left(\frac{r_i}{p_j}\right)^{n}+\frac{a_ib_jr_i}{r_i-p_j}e_{q^\alpha}^{-p_j x}e_{q^\beta}^{\frac{1}{p_j}y}E_{q^\alpha}^{r_i x}E_{q^\beta}^{-\frac{1}{r_i}y}&\delta_{i,N+1}\left(\frac{r_i}{p_{N+1}}\right)^{n}+\frac{a_ib_{N+1}r_i}{r_i-p_{N+1}}e_{q^\alpha}^{-p_{N+1} x}e_{q^\beta}^{\frac{1}{p_{N+1}}y}E_{q^\alpha}^{r_i x}E_{q^\beta}^{-\frac{1}{r_i}y}\\
			b_je_{q^\alpha}^{-p_j x}e_{q^\beta}^{\frac{1}{p_j}y}&b_{N+1}e_{q^\alpha}^{-p_{N+1} x}e_{q^\beta}^{\frac{1}{p_{N+1}}y}
	\end{vmatrix}}.
\end{align*}
To make sure that $X_{n+1}'$ is $2$-periodic, we need to require $r_{N+1}=-p_{N+1}=\lambda_{N+1}$  in addition to $r_i=-p_i=\lambda_i$ for $i=1,\dots,N$. Finally, we obtain 
\begin{align*}
X_{0}'&=\frac{\begin{vmatrix}
			-\delta_{i,j}-\frac{a_ib_j\lambda_j}{\lambda_i+\lambda_j}e_{q^\alpha}^{\lambda_j x}e_{q^\beta}^{-\frac{1}{\lambda_j}y}E_{q^\alpha}^{\lambda_i x}E_{q^\beta}^{-\frac{1}{\lambda_i}y}&-\delta_{i,N+1}-\frac{a_ib_{N+1}\lambda_{N+1}}{\lambda_i+\lambda_{N+1}}e_{q^\alpha}^{\lambda_{N+1} x}e_{q^\beta}^{-\frac{1}{\lambda_{N+1}}y}E_{q^\alpha}^{\lambda_i x}E_{q^\beta}^{-\frac{1}{\lambda_i}y}\\
			b_jp_je_{q^\alpha}^{\lambda_j x}e_{q^\beta}^{-\frac{1}{\lambda_j}y}&-b_{N+1}\lambda_{N+1}e_{q^\alpha}^{\lambda_{N+1} x}e_{q^\beta}^{-\frac{1}{\lambda_{N+1}}y}
	\end{vmatrix}}{\begin{vmatrix}
			-\delta_{i,j}+\frac{a_ib_j\lambda_i}{\lambda_i+\lambda_j}e_{q^\alpha}^{\lambda_j x}e_{q^\beta}^{-\frac{1}{\lambda_j}y}E_{q^\alpha}^{\lambda_i x}E_{q^\beta}^{-\frac{1}{\lambda_i}y}&-\delta_{i,N+1}+\frac{a_ib_{N+1}\lambda_i}{\lambda_i+\lambda_{N+1}}e_{q^\alpha}^{\lambda_{N+1} x}e_{q^\beta}^{-\frac{1}{\lambda_{N+1}}y}E_{q^\alpha}^{\lambda_i x}E_{q^\beta}^{-\frac{1}{\lambda_i}y}\\
			b_je_{q^\alpha}^{\lambda_j x}e_{q^\beta}^{-\frac{1}{\lambda_j}y}&b_{N+1}e_{q^\alpha}^{\lambda_{N+1} x}e_{q^\beta}^{-\frac{1}{\lambda_{N+1}}y}
	\end{vmatrix}}\\
	X_{1}'&=\frac{\begin{vmatrix}
			\delta_{i,j}-\frac{a_ib_j\lambda_j}{\lambda_i+\lambda_j}e_{q^\alpha}^{\lambda_j x}e_{q^\beta}^{-\frac{1}{\lambda_j}y}E_{q^\alpha}^{\lambda_i x}E_{q^\beta}^{-\frac{1}{\lambda_i}y}&\delta_{i,N+1}-\frac{a_ib_{N+1}\lambda_{N+1}}{\lambda_i+\lambda_{N+1}}e_{q^\alpha}^{\lambda_{N+1} x}e_{q^\beta}^{-\frac{1}{\lambda_{N+1}}y}E_{q^\alpha}^{\lambda_i x}E_{q^\beta}^{-\frac{1}{\lambda_i}y}\\
			b_jp_je_{q^\alpha}^{\lambda_j x}e_{q^\beta}^{-\frac{1}{\lambda_j}y}&-b_{N+1}\lambda_{N+1}e_{q^\alpha}^{\lambda_{N+1} x}e_{q^\beta}^{-\frac{1}{\lambda_{N+1}}y}
	\end{vmatrix}}{\begin{vmatrix}
			\delta_{i,j}+\frac{a_ib_j\lambda_i}{\lambda_i+\lambda_j}e_{q^\alpha}^{\lambda_j x}e_{q^\beta}^{-\frac{1}{\lambda_j}y}E_{q^\alpha}^{\lambda_i x}E_{q^\beta}^{-\frac{1}{\lambda_i}y}&\delta_{i,N+1}+\frac{a_ib_{N+1}\lambda_i}{\lambda_i+\lambda_{N+1}}e_{q^\alpha}^{\lambda_{N+1} x}e_{q^\beta}^{-\frac{1}{\lambda_{N+1}}y}E_{q^\alpha}^{\lambda_i x}E_{q^\beta}^{-\frac{1}{\lambda_i}y}\\
			b_je_{q^\alpha}^{\lambda_j x}e_{q^\beta}^{-\frac{1}{\lambda_j}y}&b_{N+1}e_{q^\alpha}^{\lambda_{N+1} x}e_{q^\beta}^{-\frac{1}{\lambda_{N+1}}y}
	\end{vmatrix}}
\end{align*}
which together with $X_0$ and $X_1$ given by \eqref{oS1} and \eqref{oS2} give solutions to the modified $q$-sG equation \eqref{BTsG3}$\sim$\eqref{BTsG6}. If we further require $a_i$ are real and $b_j$ for $i,j=1,\dots,N,N+1$ are pure imaginary, then we can draw the conclusion that $X'=X_0'=\bar{X_1'}$. Notice $X=X_0=\bar{X_1}$ obtained earlier, we finally get solutions to the reduced $q$-sG equation \eqref{rBT1}$\sim$\eqref{rBT2} . 

\section{Concluding remarks}
In literature, Casortian solutions, bilinear B$\ddot{a}$cklund transformation and Lax pair were presented for the bilinear $q$-2DTL equation by Hirota's bilinear method \cite{KOS,LNS}. Moreover, Darboux transformation was established to construct quasi-Casoratian solutions for a noncommutative $q$-2DTL equation \cite{LNS}. In this paper, we successfully derive a generalized bilinear B$\ddot{a}$cklund transformation for the bilinear $q$-2DTL equation, which reduces to the bilinear B$\ddot{a}$cklund transformation for the well-known bilinear 2DTL equation appearing in \cite{HR1,HR2} by taking the continuum limit $q\rightarrow 1$. Then a generalized Lax pair is derived along the line. As a matter of fact, the existing Darboux transformation for the noncommutative $q$-2DTL equation also works for the commutative $q$-2DTL equation \eqref{cbQT}, by which Casoratian solutions to both the bilinear $q$-2DTL equation and its bilinear B$\ddot{a}$cklund transformation are re-constructed. This reveals the profound relations between Darboux transformation and Hirota's method. As a remaining challenging problem, we successfully construct the binary Darboux transformation for the $q$-2DTL equation \eqref{cbQT}, by which Grammian solutions to the bilinear $q$-2DTL equation and its bilinear B$\ddot{a}$cklund transformation are obtained. What's more, these solutions are not only proved both by induction and direct verifications but also expressed in terms of quantum integrals. As the $2$-periodic reductions of the $q$-2DTL equation \eqref{cbQT} and its solutions, a $q$-difference sine-Gordon equation and its modified system are reported for the first time together with their corresponding solutions. 
As we know, sine-Gordon equation is of great research interest and extensively studied in literature. Therefore, we believe it is interesting to explore other properties and potential applications of the $q$-sG equation and modified $q$-sG equation in future. In addition, it is known that the $3$-periodic reduction of the 2DTL equation yields Tzitzeica equation. The $3$-periodic reductions of the $q$-2DTL equation are expected to produce something interesting.

\section*{Acknowledgement}
The authors would like to show their heartfelt gratitude to Professor Xing-Biao Hu, Professor Qing-Ping Liu and Dr. Kai Tian for their kind guidance and help. This work was supported by
National Natural Science Foundation of China (Grant Nos. 11971322, 12171475 and 11871336).

\bibliography{refs}

\end{document}